\documentclass[pra,aps,nopacs,onecolumn,twoside,superscriptaddress]{revtex4}

\usepackage{multirow, makecell}
\usepackage{amsmath,amsfonts,amssymb,caption,color,epsfig,graphics,graphicx,latexsym,mathrsfs,revsymb,theorem,url,verbatim,epstopdf,enumerate,float,mathdots}
\usepackage{fontenc}
\usepackage{cases}
\usepackage{hyperref}
\usepackage{cleveref}
\hypersetup{colorlinks,linkcolor={blue},citecolor={blue},urlcolor={red}}

\newtheorem{definition}{Definition}
\newtheorem{proposition}[definition]{Proposition}
\newtheorem{lemma}[definition]{Lemma}

\newtheorem{theorem}[definition]{Theorem}
\newtheorem{corollary}[definition]{Corollary}
\newtheorem{conjecture}[definition]{Conjecture}

\newtheorem{remark}[definition]{Remark}
\newtheorem{example}[definition]{Example}
\newtheorem{question}[definition]{Question}
\newtheorem{memo}[definition]{Memo}

\def\squareforqed{\hbox{\rlap{$\sqcap$}$\sqcup$}}
\def\qed{\ifmmode\squareforqed\else{\unskip\nobreak\hfil
		\penalty50\hskip1em\null\nobreak\hfil\squareforqed
		\parfillskip=0pt\finalhyphendemerits=0\endgraf}\fi}
\def\endenv{\ifmmode\;\else{\unskip\nobreak\hfil
		\penalty50\hskip1em\null\nobreak\hfil\;
		\parfillskip=0pt\finalhyphendemerits=0\endgraf}\fi}
\newenvironment{proof}{\noindent \textbf{{Proof.~} }}{\qed}
\def\Dbar{\leavevmode\lower.6ex\hbox to 0pt
	{\hskip-.23ex\accent"16\hss}D}
\makeatletter
\def\url@leostyle{%
	\@ifundefined{selectfont}{\def\UrlFont{\sf}}{\def\UrlFont{\small\ttfamily}}}
\makeatother
\def\bcj{\begin{conjecture}}
	\def\ecj{\end{conjecture}}
\def\bcr{\begin{corollary}}
	\def\ecr{\end{corollary}}
\def\bd{\begin{definition}}
	\def\ed{\end{definition}}
\def\bea{\begin{eqnarray}}
	\def\eea{\end{eqnarray}}
\def\bem{\begin{enumerate}}
	\def\eem{\end{enumerate}}
\def\bex{\begin{example}}
	\def\eex{\end{example}}
\def\bim{\begin{itemize}}
	\def\eim{\end{itemize}}
\def\bl{\begin{lemma}}
	\def\el{\end{lemma}}
\def\bma{\begin{bmatrix}}
	\def\ema{\end{bmatrix}}
\def\bpf{\begin{proof}}
	\def\epf{\end{proof}}
\def\bpp{\begin{proposition}}
	\def\epp{\end{proposition}}
\def\bqu{\begin{question}}
	\def\equ{\end{question}}
\def\br{\begin{remark}}
	\def\er{\end{remark}}
\def\bt{\begin{theorem}}
	\def\et{\end{theorem}}
\def\bmm{\begin{memo}}
	\def\emm{\end{memo}}

\def\btb{\begin{tabular}}
	\def\etb{\end{tabular}}

	\newcommand{\nc}{\newcommand}
	
	\nc{\bbA}{\mathbb{A}} \nc{\bbB}{\mathbb{B}} \nc{\bbC}{\mathbb{C}}
	\nc{\bbD}{\mathbb{D}} \nc{\bbE}{\mathbb{E}} \nc{\bbF}{\mathbb{F}}
	\nc{\bbG}{\mathbb{G}} \nc{\bbH}{\mathbb{H}} \nc{\bbI}{\mathbb{I}}
	\nc{\bbJ}{\mathbb{J}} \nc{\bbK}{\mathbb{K}} \nc{\bbL}{\mathbb{L}}
	\nc{\bbM}{\mathbb{M}} \nc{\bbN}{\mathbb{N}} \nc{\bbO}{\mathbb{O}}
	\nc{\bbP}{\mathbb{P}} \nc{\bbQ}{\mathbb{Q}} \nc{\bbR}{\mathbb{R}}
	\nc{\bbS}{\mathbb{S}} \nc{\bbT}{\mathbb{T}} \nc{\bbU}{\mathbb{U}}
	\nc{\bbV}{\mathbb{V}} \nc{\bbW}{\mathbb{W}} \nc{\bbX}{\mathbb{X}}
	\nc{\bbZ}{\mathbb{Z}}
	
	\nc{\bA}{{\bf A}} \nc{\bB}{{\bf B}} \nc{\bC}{{\bf C}}
	\nc{\bD}{{\bf D}} \nc{\bE}{{\bf E}} \nc{\bF}{{\bf F}}
	\nc{\bG}{{\bf G}} \nc{\bH}{{\bf H}} \nc{\bI}{{\bf I}}
	\nc{\bJ}{{\bf J}} \nc{\bK}{{\bf K}} \nc{\bL}{{\bf L}}
	\nc{\bM}{{\bf M}} \nc{\bN}{{\bf N}} \nc{\bO}{{\bf O}}
	\nc{\bP}{{\bf P}} \nc{\bQ}{{\bf Q}} \nc{\bR}{{\bf R}}
	\nc{\bS}{{\bf S}} \nc{\bT}{{\bf T}} \nc{\bU}{{\bf U}}
	\nc{\bV}{{\bf V}} \nc{\bW}{{\bf W}} \nc{\bX}{{\bf X}}
	\nc{\bZ}{{\bf Z}}
	
	\nc{\as}{{\cal AS}}
	\nc{\app}{{\cal AP}}
	\nc{\ar}{{\cal AR}}
	\nc{\bp}{{\cal BP}}
	\nc{\dbp}{{\cal DBP}}
\nc{\ew}{{\cal EW}}
\nc{\dew}{{\cal DEW}}
\nc{\ndew}{{\cal NDEW}}
\nc{\conv}{{\text{Conv}}}

	\nc{\cA}{{\cal A}} \nc{\cB}{{\cal B}} \nc{\cC}{{\cal C}}
	\nc{\cD}{{\cal D}} \nc{\cE}{{\cal E}} \nc{\cF}{{\cal F}}
	\nc{\cG}{{\cal G}} \nc{\cH}{{\cal H}} \nc{\cI}{{\cal I}}
	\nc{\cJ}{{\cal J}} \nc{\cK}{{\cal K}} \nc{\cL}{{\cal L}}
	\nc{\cM}{{\cal M}} \nc{\cN}{{\cal N}} \nc{\cO}{{\cal O}}
	\nc{\cP}{{\cal P}} \nc{\cQ}{{\cal Q}} \nc{\cR}{{\cal R}}
	\nc{\cS}{{\cal S}} \nc{\cT}{{\cal T}} \nc{\cU}{{\cal U}}
	\nc{\cV}{{\cal V}} \nc{\cW}{{\cal W}} \nc{\cX}{{\cal X}}
	\nc{\cZ}{{\cal Z}}
	
	\nc{\cpp}{{\cal PP}}
	
	\nc{\hA}{{\hat{A}}} \nc{\hB}{{\hat{B}}} \nc{\hC}{{\hat{C}}}
	\nc{\hD}{{\hat{D}}} \nc{\hE}{{\hat{E}}} \nc{\hF}{{\hat{F}}}
	\nc{\hG}{{\hat{G}}} \nc{\hH}{{\hat{H}}} \nc{\hI}{{\hat{I}}}
	\nc{\hJ}{{\hat{J}}} \nc{\hK}{{\hat{K}}} \nc{\hL}{{\hat{L}}}
	\nc{\hM}{{\hat{M}}} \nc{\hN}{{\hat{N}}} \nc{\hO}{{\hat{O}}}
	\nc{\hP}{{\hat{P}}} \nc{\hR}{{\hat{R}}} \nc{\hS}{{\hat{S}}}
	\nc{\hT}{{\hat{T}}} \nc{\hU}{{\hat{U}}} \nc{\hV}{{\hat{V}}}
	\nc{\hW}{{\hat{W}}} \nc{\hX}{{\hat{X}}} \nc{\hZ}{{\hat{Z}}}
	
	\nc{\hn}{{\hat{n}}}

	\def\Dbar{\leavevmode\lower.6ex\hbox to 0pt
		{\hskip-.23ex\accent"16\hss}D}

\begin{document}

\title{Commutators with multiple unitary symmetry}

\author{Shu Li}\email{leo7090@buaa.edu.cn}
\affiliation{LMIB(Beihang University), Ministry of education, and School of Mathematical Sciences, Beihang University, Beijing 100191, China}
\author{Jie Wang}\email{22377088@buaa.edu.cn}
\affiliation{LMIB(Beihang University), Ministry of education, and School of Mathematical Sciences, Beihang University, Beijing 100191, China}
\author{Binfeng Wang}\email[]{22377078@buaa.edu.cn}
\affiliation{LMIB(Beihang University), Ministry of education, and School of Mathematical Sciences, Beihang University, Beijing 100191, China}
\author{Lin Chen}\email[]{linchen@buaa.edu.cn(corresponding author)}
\affiliation{LMIB(Beihang University), Ministry of education, and School of Mathematical Sciences, Beihang University, Beijing 100191, China}

\date{\today}

\begin{abstract}
Commutators are essential in quantum information theory, influencing quantum state symmetries and information storage robustness. This paper systematically investigates the characteristics of bipartite and multipartite quantum states invariant under local unitary group actions. The results demonstrate that 
 any quantum states commuting with $U \otimes U^{\dagger}$ and $U \otimes V$ can be expressed as $\frac{1}{n}I_n$, where $U$ and $V$ are arbitary $n\times n$ unitary matrices. Furthermore, in tripartite systems, any quantum states commuting with $U \otimes U \otimes U^{\dagger}$ must necessarily adopt the form:
$W = xI_{n^3} + y\left(\sum_{i,j=1}^n (|i\rangle \langle j|) \otimes (|j\rangle \langle i|)\right) \otimes I_n$,
where $F_n$ represents the canonical swap operator. These results provide theoretical tools for characterizing multipartite entanglement constraints and designing symmetry-protected quantum protocols.
\end{abstract}
\maketitle

\section{Introduction}

In the fields of quantum information theory and linear algebra, the study of symmetry \cite{marvian2012symmetry,zeng2019quantum} has always been a core topic, especially concerning the tensor products of unitary matrices \cite{ryan2002introduction,keyl2002fundamentals} and their transformation properties. These issues are of significant importance in the analysis of quantum state symmetry, the identification of quantum entanglement \cite{horodecki2009quantum,bengtsson2017geometry}, and the study of invariance in quantum systems \cite{ticozzi2008quantum}, attracting the attention of numerous scholars.

Werner and Reinhard \cite{werner1989quantum} first introduced quantum states with $U \otimes U$ symmetry, focusing on characterizing the set of positive semi-definite matrices $W$ commuting with $U \otimes U$. They proved that such matrices can be simplified as $W = xI_{n^2} + yF_n$, where $F_n = \sum_{i,j=1}^n (|i\rangle\langle j|) \otimes (|j\rangle\langle i|)$, establishing a foundational framework for subsequent studies.

Subsequent works extended these investigations: Eggeling and Werner \cite{eggeling2001separability} analyzed the separability of tripartite $U^{\otimes 3}$-symmetric states, revealing the geometric structures of fully separable, biseparable, and positive partial transpose (PPT) states. They demonstrated that PPT conditions become less restrictive in higher dimensions and identified a nontrivial separation between tripartite entanglement and biseparable states. Johnson \cite{johnson2014state} developed an analytical framework for local positive connectivity of quantum channels in tripartite systems, achieving precise characterizations in low-dimensional cases. Leveraging Schur-Weyl duality, Goodman and Wallach \cite{goodman2009symmetry} showed that all $U(n)$-invariant operators decompose into irreducible representations of the symmetric group $S(n)$. Jafarizadeh \cite{jafarizadeh2020study} utilized Twirl operations to prove that the density matrix of three-qubit Werner states exhibits a block-diagonal structure in the total spin representation, with eigenvectors corresponding to idempotents of the $S_3$ correlation scheme.

The connection to the theory of irreducible representations \cite{munn1961class} has attracted significant mathematical interest. In 2013, Studzinski \cite{studzinski2013commutant} constructed irreducible representations of partially transposed permutation operators and demonstrated that when the local dimension \( d > n-2 \), the algebra is isomorphic to the walled Brauer algebra and is semisimple. Subsequent research in 2025 \cite{studzinski2025irreducible} proposed a recursive method for constructing irreducible matrix representations of symmetric group algebras, extending the approach to residual ideals and establishing a systematic algebraic framework for analyzing symmetric quantum states. This problem is also closely related to the study of multi-body local unitary equivalence. Chen \cite{chen2015universal} proved the universality of single-occupied subspaces under local unitary groups for three-fermion systems, which fails for \( N > 3 \), with counterexamples such as BCS states that cannot be transformed via local unitary equivalence. Song \cite{song2022proof} demonstrated that for any tripartite mixed state \( \rho_{ABC} \), the ranks of the reduced density matrices satisfy the inequality \( r(\rho_{AB}) \cdot r(\rho_{AC}) \geq r(\rho_{BC}) \), and further generalized this result to multipartite quantum systems, including conditions for equality. Shi \cite{shi2025entanglement} introduced the concept of entanglement detection length (EDL) for symmetric quantum states, determining it through marginal separability and providing a semidefinite programming upper bound, thereby revealing the maximal discrepancy between EDL and state determination length.

However, existing research has not yet addressed the conjugate transpose of unitary matrices. In quantum physics, matrix transposition is a common operation, and its practical significance can be represented as a mapping, channel, or noise, used to describe whether the initial state can remain invariant under noise interference. Based on the above discussion, this paper focuses on the general solution of the following equation,
\begin{equation}
	(U \otimes U^\dagger) W (U \otimes U^\dagger)^\dagger = W
	\label{original}.
\end{equation}

For convenience, our $W$ is not necessarily normalized; however, as a quantum state, it must be interpreted as a normalized positive semi-definite matrix. We employ a matrix partitioning approach to divide the $n^2 \times n^2$ matrix $W$ into $n \times n$ block matrices. By utilizing the matrix form of Schur's lemma and the arbitrariness of $U$, we solve the problem through a divide-and-conquer strategy: diagonal blocks and off-diagonal blocks are separately investigated in Lemma \ref{le:} and \ref{diag1}, followed by merging the results to conclude that quantum state $W$ must be expressed as $cI_n(c\geq 0)$. Based on this, we relax the constraints on the matrices and provide solutions to the equation for several specific cases.

Subsequently, we conduct further analyses on $U \otimes V$ and $U \otimes U \otimes U^{\dagger}$. By selecting matrices with distinct eigenvalues, we demonstrate that the commutant of $U \otimes V$ consists solely of scalar matrices in Theorem \ref{uv}. For $U \otimes U \otimes U^{\dagger}$, according to Werner's result \cite{werner1989quantum}, $W=xI_{n^3} + yF_n \otimes I_n$
constitutes partial solutions to the equation. We continue to adopt the divide-and-conquer strategy. We first deduce that solving $W$ is equivalent to simplifying maatrix block $W_{1,1}$ and $W_{1,2}$ in Theorem \ref{1112}.By substituting permutation matrices and special unitary matrices into the equation in Lemma \ref{simplify-1} to Theorem \ref{1112=0}, we prove that the aforementioned expression must be the complete set of solutions.

This study provides a unified framework for analyzing high-dimensional matrix structures, which can facilitate advancements in quantum information theory, tensor decomposition algorithms. The rest of this paper is organized as follows. In Sec. \ref{sec:pre} we introduce the preliminary knowledge used in this paper. We present our main findings in Sec. \ref{sec:commutant} and \ref{sec:Uotimes V}. Finally we conclude in Sec. \ref{sec:con}.

\section{Preliminaries}
\label{sec:pre}

We introduce some basic facts and definitions used in this article. We assume that $\mathbb{C}^d$ is the $d$-dimensional Hilbert space. We refer to $I_d$ as the $d \times d$ unit matrix, $P_{ij}$ as the permutation matrix obtained by exchanging the $i-th$ and $j-th$ rows of $I_n$, and $E_{ij}$ as the matrix with 1 in the $(i,j)$-th entry and 0 elsewhere. We also denote $A^T,A^\dagger,A^*$ as the transposition, conjugate transposition,complex conjugation of the matrix $A$ respectively. We begin with the simplest unitary matrices.
\begin{lemma}
Every $2\times 2$ unitary matrices can be written with the following general form, where \( \theta \in \left[0, \frac{\pi}{2}\right] \), \( \alpha \in [0, 2\pi] \), \( \beta \in [0, 2\pi] \), and \( \varphi \in [0, 2\pi] \). The determinant of \( U \) equals to \( e^{i\varphi} \).
\begin{equation}
U=\left( \begin{array}{cc} 
e^{i\alpha}\cos\theta & -e^{i(\varphi-\beta)}\sin\theta \\ 
e^{i\beta}\sin\theta & e^{i(\varphi-\alpha)}\cos\theta 
\end{array} \right)
=\left( \begin{array}{cc} 
e^{-i\beta} & 0 \\ 
0 & e^{-i\alpha} 
\end{array} \right)
\left( \begin{array}{cc} 
\cos\theta & -\sin\theta \\ 
\sin\theta & \cos\theta 
\end{array} \right)
\left( \begin{array}{cc} 
e^{i(\alpha+\beta)} & 0 \\ 
0 & e^{i\varphi} 
\end{array} \right).
\end{equation}
\label{second-unitary}
\qed
\end{lemma}
The content of this study is also closely related to the commutativity of matrices. Now we introduce two Lemmas related to this topic.
\begin{lemma}
Let \( A \) and \( B \) be $n \times n$ square matrices, where \( A \) has \( n \) distinct eigenvalues, and \( AB = BA \). Then there exists a polynomial \( f(x) \) of degree not exceeding \( n-1 \) such that \( B = f(A) \).
\label{commute1}
\end{lemma}
\begin{lemma}
The matrix that commutes with all $n \times n$ unitary matrices is a scalar matrix $cI_n, c\in \mathbb{C}$.
\label{commute2}
\end{lemma}
Next, we introduce some basic properties on Kronecker product used in this paper. 
\begin{lemma}
    (1)If $A,B,C,D$ are matrices of such size that can form matrix products $AC$ and $BD$, then $(A\otimes B)(C\otimes D)=(AC\otimes BD)$;\\
    (2)Suppose that $A$ is a $n\times n$ matrix that eigenvalues are $\lambda_1,\cdots,\lambda_n$, $B$ is a $m\times m$ matrix that eigenvalues are $\mu_1,\cdots,\mu_m$, then the eigenvalues of $A\otimes B$ are $\lambda_i\mu_j,i=1,\cdots,n,j=1,\cdots,m$.
    \end{lemma}
\begin{lemma}
If the matrix $A$ commutes with $W_1$, and the matrix $B$ commutes with $W_2$, then the matrix  $A\otimes B$ commutes with $W_1\otimes W_2$.
    \label{commute3}
\end{lemma}
\begin{lemma}
Let $F_n = \sum\limits_{i, j=1}^{n} (|i\rangle \langle j|) \otimes (|j\rangle \langle i|)$. If we divide $F_n$ into $n^2$ $n\times n$ matrix blocks, then the $(i,j)-th(i\neq j)$ block of $F_n$ is $E_{ji}$; the $(i,i)-th$ block of $F_n$ is $E_{ii}$.\\
\label{centrosymmetric}
\end{lemma}

\section{Commutant structure for $U\otimes U^\dagger$ transformations}
\label{sec:commutant}

In this section, we focus on the equation \eqref{original}. We partition the $n^2\times n^2$ positive semidefinite matrix $W$ into $n^2$ blocks $n\times n$ matrices, and label them as $W_{ij}, i,j\in \{1,\cdots n\}.$ We first demonstrate a lemma about off-diagonal matrix blocks $W_{i,j}$ with $i\ne j$. We shall refer to $P_{i,j}$ as the permutation matrix switching the matrix elements $(i,j)$ and $(j,i)$.

\begin{lemma}
\label{le:}
Every matrix $ W_{i,j}(i\neq j)$ can be expressed as $P(i,j)W_{1,2}P^T(i,j)$ , where $P(i,j)$ is the products of some permutation matrices.
\end{lemma}
\begin{proof}
We prove the claim by induction. Suppose for some integer $k\ge2$, we have  $W_{i,j}=P(i,j)W_{1,2}P(i,j)^T,$where $P(i,j)$ is the products of some permutation matrices. We prove the claim for $k+1$. Setting $U=P_{k,k+1}$ in \eqref{original}, we have
\begin{equation*}
    W_{k,j}=P(k,j)W_{1,2}P(k,j)^T=P_{k,k+1}W_{k+1,j}P_{k,k+1},\\
    W_{j,k}=P(j,k)W_{1,2}P(j,k)^T=P_{k,k+1}W_{j,k+1}P_{k,k+1},\,j=1,\cdots,k-1.
\end{equation*} So $W_{i,k+1},W_{k+1,j},\forall i\neq j(i,j\leq k-1)$ can be expressed as the product of $W_{1,2}$ and some permutation matrices. For the remaining $W_{k,k+1},W_{k+1,k}$, we choose $U=P_{1,k+1}$ and obtain $P(k,1)W_{1,2}P(k,1)^T=P_{1,k+1}W_{k,k+1}P_{1,k+1}$, as well as $P(1,k)W_{1,2}P(1,k)^T=P_{1,k+1}W_{k+1,k}P_{1,k+1}$. So the induction holds.
\end{proof}

Next we investigate two cases, namely diagonal matrix blocks and off-diagonal matrix blocks in \eqref{original}. This leads to Lemmas \ref{diag2} and \ref{diag1}, respectively.

\begin{lemma}
If $W$ satisfies \eqref{original}, then $W_{i,j}=\mathbf{O}, \forall i\neq j$.
\label{diag2}
\end{lemma}
\begin{proof}
We need show the relation between the blocks $W_{12}$ and $W_{21}$, as well as the relation between the blocks $W_{12}$ and $W_{13}$. That is, we need find the permutation matrices $P_{ij}$ for the two relations in Lemma 4. First we choose $U=P_{1,2}$ in equation \eqref{original} and obtain
\begin{equation}
    W=\begin{pmatrix}
        \mathbf{O} & P_{1,2} & \mathbf{O}&\mathbf{O}  \\
        P_{1,2}& \mathbf{O} & \mathbf{O}&\mathbf{O}\\
        \mathbf{O}& \mathbf{O}& P_{1,2}& \mathbf{O}\\
        \mathbf{O}& \mathbf{O}& \mathbf{O}&  \ddots\\
    \end{pmatrix}\begin{pmatrix}
        W_{1,1} & W_{1,2} & * & * \\
        W_{2,1}& W_{2,2} & *& *\\
        *& *& *& *&\\
        *& *& *& \ddots&\\
    \end{pmatrix}\begin{pmatrix}
        \mathbf{O} & P_{1,2} & \mathbf{O}&\mathbf{O}  \\
        P_{1,2}& \mathbf{O} & \mathbf{O}&\mathbf{O}\\
        \mathbf{O}& \mathbf{O}& P_{1,2}& \mathbf{O}\\
        \mathbf{O}& \mathbf{O}& \mathbf{O}&  \ddots\\
    \end{pmatrix}=\begin{pmatrix}
       *& P_{1,2}W_{2,1}P_{1,2} & * &*  \\
        P_{1,2}W_{1,2}P_{1,2}& * & *&* \\
        *& *&*&* &\\
        *& *&* &\ddots\\
    \end{pmatrix}.
\label{e12_pre}
\end{equation}
So $W_{2,1}=P_{1,2}W_{1,2}P_{1,2}$. Similarly, let $U=P_{2,3}$, we can get $W_{1,3}=P_{2,3}W_{1,2}P_{2,3}$
Next, let $U=U_1=\begin{pmatrix}
        0 & 1 &  & & \\
        i& 0 & && \\
        & &1 & & \\
        &&&\ddots& \\
        &&& & 1\\
    \end{pmatrix},U=U_2=\begin{pmatrix}
        0 & i &  & & \\
        1& 0 & && \\
        & &1 & & \\
        &&&\ddots& \\
        &&& & 1\\
    \end{pmatrix}$ in equation \ref{original} respectively. Let $W_{1,2}=(b_{ij})_{n\times n}$, comparing  the (1,2)-blocks of $(U \otimes U^\dagger) W (U \otimes U^\dagger)^\dagger$  and $W$ we obtain
\[\footnotesize
\begin{array}{rcl}
W_{1,2} &=& \begin{pmatrix}
b_{11} & b_{12} & b_{13} & \cdots & b_{1n} \\
b_{21} & b_{22} & b_{23} & \cdots & b_{2n} \\
b_{31} & b_{32} & b_{33} & \cdots & b_{3n} \\
\vdots & \vdots & \vdots & \ddots & \vdots \\
b_{n1} & b_{n2} & b_{n3} & \cdots & b_{nn} \\
\end{pmatrix} \\
&\overset{U_1}{=}& -iU_1^{\dagger}P_{1,2}W_{1,2}P_{1,2}U_1^{\dagger} = \begin{pmatrix}
-ib_{11} & -b_{12} & -b_{13} & \cdots & -b_{1n} \\
b_{21} & -ib_{22} & -ib_{23} & \cdots & -ib_{2n} \\
b_{31} & -ib_{32} & -ib_{33} & \cdots & -ib_{3n} \\
\vdots & \vdots & \vdots & \ddots & \vdots \\
b_{n1} & -ib_{n2} & -ib_{n3} & \cdots & -ib_{nn} \\
\end{pmatrix} \\
&\overset{U_2}{=}& iU_2^{\dagger}P_{1,2}W_{1,2}P_{1,2}U_2^{\dagger} = \begin{pmatrix}
ib_{11} & -b_{12} & ib_{13} & \cdots & ib_{1n} \\
b_{21} & ib_{22} & b_{23} & \cdots & b_{2n} \\
ib_{31} & -b_{32} & ib_{33} & \cdots & ib_{3n} \\
\vdots & \vdots & \vdots & \ddots & \vdots \\
ib_{n1} & -b_{n2} & -ib_{n3} & \cdots & ib_{nn} \\
\end{pmatrix}
\end{array}
\]
We deduce that only $b_{21}$ in $W_{1,2}$ may be nonzero. Next, taking $U=U_3=\begin{pmatrix}
         &  &1  & & \\
        & 1 & && \\
        i& & & & \\
        &&&1& \\
        &&&& \ddots\\
        &&& & &1\\
    \end{pmatrix}$ and still computing the (1,2)-blocks similarly, we obtain $b_{21}=0$, thereby leading to $W_{1,2}=\mathbf{O}$, then $W_{i,j}=\mathbf{O},\forall i\neq j$.    
\end{proof}
\begin{lemma}
If $W$ satisfies equation (\ref{original}), then $W_{i,i}=cI_{n},\forall i\in \{1,\cdots n\}$,$\forall c\in \mathbb{C}$.
\label{diag1}
\end{lemma}
\begin{proof}
Let \( U = U_n = \begin{pmatrix}
        1 &   \\
        & U_{n-1} \\
    \end{pmatrix} \) in equation \ref{original} (where \( U_{n-1} \) is an arbitrary \((n-1) \times (n-1)\) unitary matrix), we obtain \( U_n^\dagger W_{1,1} U_n = W_{1,1} \).

Specifically, taking \( U_n = \text{diag}\{1, e^{i\theta_1}, \cdots, e^{i\theta_{n-1}}\} \) with \( \theta_j \neq \theta_k \) for all \( j \neq k \) and \( \theta_i \neq 2k\pi \), we can deduce that \( W_{1,1} \) is a diagonal matrix according to Lemma \ref{commute1}.  Let $W_{1,1}=\begin{pmatrix}
        a &   \\
        & W_0 \\
    \end{pmatrix}$, then we also have $U_{n-1}^\dagger W_0U_{n-1}=W_0$. According to the randomness of $U_{n-1}$, applying lemma \ref{commute2},we get  $W_{1,1}=\begin{pmatrix}
        a &   \\
        & cI_{n-1} \\
    \end{pmatrix}$. Next, letting $U=P_{1,2},\cdots P_{1,n}$ 
 in equation \ref{original} successively,,we have $W_{1,1}=P_{1,j}W_{j,j}P_{1,j},\forall j\in\{2,\cdots,n\}$ That means that $W_{j,j}=diag\{c,\cdots,c,a,c,\cdots,c\}$ where the $j^th$ diagonal element of $W_{j,j}$ is a.Furthermore, we choose U as a unitary matrix $U_0$ with the first row entirely composed of $\frac{1}{\sqrt{n}}$,we have
 $$W_{1,1}=((U_0 \otimes U_0^\dagger) W (U_0 \otimes U_0^\dagger)^\dagger)_{1,1}=\frac{1}{n}U_0^\dagger(\sum\limits_{i=1}^nW_{i,i})U_0=\frac{1}{n}(a+(n-1)c)I_n,$$
 where the second to last equation mark derives from Lemma \ref{diag2}. Thus $a=c,W_{i,i}=cI_{n},\forall i\in \{1,\cdots n\}.$
\end{proof}

Combining Lemmas \ref{diag2} and \ref{diag1}, we conclude the following fact.
\begin{theorem}
If $W$ satisfies $(U \otimes U^\dagger) W (U \otimes U^\dagger)^\dagger = W$ for all $n\times n$ unitary matrices $U$, then  $W=cI_{n^2}$ with a positive number $c$.
\label{thm:diag2}
\end{theorem}
Now we relax the condition that $U$ is a unitary matrix and continue to discuss the existence of solutions to the equation \eqref{original}. We begin with the case where $U$ is a $2\times 2$ orthogonal matrix in the following observation.
\begin{theorem} 
For all  $2\times 2$ orthogonal matrices \( U \), the following matrix is a positive semidefinite solution to \eqref{original}. 
\begin{equation}
    W = xI_4 + y M\otimes M( x > |y|),M=\begin{bmatrix}
    0&1\\
    -1&0
\end{bmatrix}.
\label{2-ortho}
\end{equation} 
\end{theorem}
\begin{proof}
Let $U=\begin{bmatrix} 0&1\\ 1&0\end{bmatrix},\begin{bmatrix} 1&0\\ 0&-1\end{bmatrix},\frac{1}{\sqrt{2}}\begin{bmatrix} 1&-1\\ 1&1\end{bmatrix},\frac{1}{\sqrt{2}}\begin{bmatrix} 1&1\\ 1&-1\end{bmatrix}$, respectively. Using straightforward computation, one can show that $W$ can be written in the form of \eqref{2-ortho}. Especially for $M\otimes M$, since M satisfies $M^T=-M$, when $U$ is a rotation matrix, we have $UMU^T=U^TMU=M$; when $U$ is a reflection matrix, we have $UMU^T=U^TMU=-M$; thus we have $(U \otimes U^T)(M\otimes M)(U \otimes U^T)^T = (UMU^T)\otimes(U^TMU)=M\otimes M
$. Since both $I_4$ and $M\otimes M$ are solutions of equation \eqref{original}, and hence \eqref{2-ortho} is the general solution form of the equation.   
\end{proof}
\section{Applications by \( U \otimes V \) and \( U \otimes U \otimes U^\dagger \)}
\label{sec:Uotimes V}

In this section, we primarily explore several variants of the original equation \eqref{original}, focusing on the case where the Kronecker product of two matrices with different dimensions is considered. We first demonstrate the following fact.
It means that the commutator $W$ is trivial in such a case.
\begin{theorem} 
Let $U$ and $V$ respectively be $m\times m$ and $n\times n$ unitary matrices. Then, the  positive semi-definite solution $W$ satisfying the equation $ (U \otimes V)W(U \otimes V)^\dagger=W$ is a scalar matrix $cI_{mn}$ with some $c\geq 0$.
\label{uv}
\end{theorem}
\begin{proof} 
Let $U_1=\begin{pmatrix}
        e^{i \frac{2\pi}{m}} & & & &\\
        &\ddots& & &\\
        & &e^{i \frac{2k\pi}{m}}& &\\
        & & &\ddots& &\\
        & & & & 1&\\
\end{pmatrix}, V_1=\begin{pmatrix}
        e^{i \frac{2\pi}{a}} & & & &\\
        &\ddots& & &\\
        & &e^{i \frac{2k\pi}{a}}& &\\
        & & &\ddots& &\\
        & & & & e^{i \frac{n\pi}{a}}&\\
\end{pmatrix}$, where $a>n$ and $\frac{a}{m}$ is irrational. Then the eigenvalues of $U_1 \otimes V_1$ are $e^{i (\frac{2\pi j}{m}+\frac{2\pi k}{a})},\forall j\in\{1,\cdots ,m\},k\in\{1,\cdots ,n\}$. It is evident that each eigenvalue is simple: If there exists an eigenvalue with algebraic multiplicity greater than 1, then $\exists j_1,k_1,j_2,k_2\in \mathbb{N}\quad s.t. \frac{2\pi j_1}{m}+\frac{2\pi k_1}{a}=\frac{2\pi j_2}{m}+\frac{2\pi k_2}{a}, a.e.\frac{a}{m}=\frac{k_2-k_1}{j_1-j_2}\in \mathbb{Q},$  which would imply a contradiction.\\
Applying Lemma \ref{commute1}, there exists a polynomial \( p(x) \) of degree at most \( mn - 1 \) such that \(W=p(U \otimes V) \). From this, it follows that \( W \) is a diagonal matrix.\\
Now let \( V = I_n \) and let \( U \) traverse all \( m \)-th order permutation matrices; then let \( U = I_m \) and let \( V \) traverse all $n \times n$ permutation matrices. By combining these results, we obtain \( w_{ii}=w_{jj},\forall i,j\in \{1,\cdots,mn\} \), and thus \( W=cI_{mn} \). Since \( W \) is a positive semidefinite matrix, it follows that $cI_{mn}(c\geq 0)$.
\end{proof}

It should be noted that the most crucial prerequisite for the validity of theorem \ref{uv} is the feasibility to find $U\otimes V$ which has $n$ distinct eigenvalues. This is not achievable in the cases of $U\otimes U$ and $U\otimes U^\dagger$.Next we discuss the commutant structure of the operator \( U \otimes U \otimes U^\dagger \), specifically aiming to find a positive semifinite matrix W that satisfies the equation for every unitary matrix U, which is given by
\begin{equation}
    (U \otimes U \otimes U^\dagger)W(U \otimes U \otimes U^\dagger)^\dagger=W.
\label{3d}
\end{equation}

The former study has shown the general form of $W$ that commutes with $U\otimes U$ \cite{werner1989quantum}. Using Lemma \ref{commute3}, we can derive partial solutions to the equation \eqref{3d} as follows.
\begin{theorem}
    A positive semidefinite solution to \eqref{3d} can be written as the following form, where $x, y \in \mathbb{R},\,  \text{and} \, F_n = \sum_{i, j=1}^{n} (|i\rangle \langle j|) \otimes (|j\rangle \langle i|) $.
\begin{equation}
    \tilde{W} = xI_{n^3} + yF_n\otimes I_n.
\label{3d_sol}
\end{equation}
\qed
\end{theorem}
We proceed to analyze in detail the structure of the matrix in equation \eqref{3d}. To facilitate the exposition, we adopt notation analogous to that in Section~2. For an $n^3 \times n^3$ matrix $W$, we partition it into 
 $n \times n$ block matrices $\{W_{i,j}\}_{i,j=1}^n$, where each block $W_{i,j}$ is of size $n^2 \times n^2$. We further subdivide $W_{1,1}$ and $W_{1,2}$ into $n \times n$ subblocks $\{\mathbf{M}_{ij}\}_{i,j=1}^n$ and $\{\mathbf{N}_{ij}\}_{i,j=1}^n$, each 
 $\mathbf{M}_{ij},\mathbf{N}_{i,j}$ being an $n \times n$ matrix. Let $W'$ denote the simplified form of $W$ obtained by substituting some unitary matrices in equation \eqref{3d}. We have the following theorem:

\begin{theorem}
\label{1112}
 If $W'_{1,1} = \tilde{W}_{1,1}$ and $W'_{1,2} = \tilde{W}_{1,2}$, then equation \eqref{3d_sol} constitutes the complete solution set.
\end{theorem}
\begin{proof}
    If $W_0$ is a solution of \eqref{3d},We partition it into $n \times n$ blocks {$W_{i,j}^{0}$}. Since (1,1)-block and (1,2)-block have already been simplified to $\tilde{W}_{11},\tilde{W}_{12}$, we have $W_{11}^0=t\tilde{W}_{11},W_{1,2}^{0}=t\tilde{W}_{12},t \in C$. The case that $t=0$ is trivial, so we can assume that $t\neq 0$. By employing the method of Lemma \ref{le:} and substituting \( P_{ij}\otimes P_{ij} \) for \( P_{ij} \), we can deduce that $\forall i\neq j$ , \( W_{ij}^0 = (P\otimes P)W_{12}^0(P\otimes P)  \), where $P$ is the matrix products of a series of permutation matrices.

Now we consider \(  W_{ij}^{0}, \forall i\neq j \). According to Lemma \ref{centrosymmetric}, \( W_{12}^0=t\tilde{W}_{1,2}\) has precisely \( n \) identical nonzero elements. Since each \( W_{ij}^0\) is obtained by permuting the elements of \( W_{12}^0 \),  each \(  W_{ij}^0 \) has precisely \( n \) identical nonzero elements. We state a claim that \( W_{ij}^0 = t\tilde{W}_{ij} \):\\
If this claim is invalid, due to that both \( W_{ij}^0 \) and \( \tilde{W}_{ij} \) possess only \( n \) identical nonzero elements, there must exist a position $(k,l)$, where the $(k,l)-th$ element of \( W_{ij}^0 \) is nonzero while the $(k,l)-th$ element of \( \tilde{W}_{ij} \) is 0. Meanwhile, we note that \( W^{''} = W_{0} +\tilde{W}  \) constitutes a solution to the equation, since both \( W_{0} \) and \(\tilde{W} \) are solutions. At this juncture, the off-diagonalmatrix blocks of \( W^{''} \) satisfy \(W^{''}_{ij} =W_{ij}^0+\tilde{W}_{ij}  \), and the number of the nonzero elements of \( W^{''}_{ij} \) can never be \( n \), which leads to the contradiction.Thus \( W_{ij}^0=t\tilde{W}_{ij} \).
\\Let $U=P_{ij},1\leq i,j \leq n$ in equation \eqref{3d}, we have \( W_{ii}^0 =(P_{1j}\otimes P_{1j})W_{11}^0(P_{1j}\otimes P_{1j})   \).Since \( W_{11}^{'}= \tilde{W}_{11}\) ,\( W_{11}^0=t\tilde{W}_{11} \), we have \( W_{ii}^0 =(P_{1i}\otimes P_{1i})W_{11}^0(P_{1i}\otimes P_{1i})   \)=$t (P_{1i}\otimes P_{1i})\tilde{W}_{11}(P_{1i}\otimes P_{1i})$=t$\tilde{W}_{ii}$,where the last equal sign is derived by lemma \eqref{centrosymmetric}. Combine the results $W_{ii}^0 =\tilde{W}_{ii}$ and $W_{ij}^0 = t\tilde{W}_{ij}(\forall i\neq j)$ together,we have $W_0=t\tilde{W}$,so $\tilde{W}$ constitutes the complete solution set to equation \eqref{3d}.
\end{proof}
\begin{lemma}
\label{simplify-1}
The (1,2) $n^2\times n^2$ matrix block \( W_{1,2} \) can be reduced to a form where only the (2,1)-block \( \mathbf{N}_{2,1} \) is nonzero, with \( \mathbf{N}_{2,1} =  \begin{pmatrix}
        m_{11}& &   \\
       & m_{22} & \\
        &  & m_{33}\\
    \end{pmatrix}\oplus\begin{pmatrix}
        m_{44}& \cdots& m_{4n}  \\
        \vdots& \vdots & \vdots\\
        m_{n4}&  \cdots& m_{nn}\\
    \end{pmatrix}\).
\label{3d-nondiag}
\end{lemma}
\begin{proof}
   Let $U=P_{12}$ in equation \eqref{3d}. Mentioning that $U\otimes U\otimes U^\dagger=\begin{pmatrix}
        \mathbf{O} & P_{1,2}\otimes P_{1,2} & \mathbf{O}&\mathbf{O}  \\
        P_{1,2}\otimes P_{1,2}& \mathbf{O} & \mathbf{O}&\mathbf{O}\\
        \mathbf{O}& \mathbf{O}& P_{1,2}\otimes P_{1,2}& \mathbf{O}\\
        \mathbf{O}& \mathbf{O}& \mathbf{O}&  \ddots\\
    \end{pmatrix}$ ,we can just replace $P_{1,2}\otimes  P_{12}$ with $P_{1,2}$ in \eqref{e12_pre} in equation \ref{e12_pre}, and get 
\begin{equation}
W_{1,2}=(P_{1,2}\otimes P_{1,2})W_{2,1}(P_{1,2}\otimes P_{1,2})=(P_{1,3}\otimes P_{1,3})W_{23}(P_{1,3}\otimes P_{1,3})
\label{3d_w12}
\end{equation}
Now let $U=\begin{pmatrix}
           & e^{i\theta_2} &   &   & \\
        e^{i\theta_1}&    &   & & \\
          &   & 1&  & \\
          &   &   &  \ddots&\\
          &   &   &  &1
    \end{pmatrix}$ in equation \eqref{3d}. For $W_{1,2}$, focusing on the (1,2)-block in $(U \otimes U \otimes U^\dagger)W(U \otimes U \otimes U^\dagger)^\dagger$ and $W$, we have 
\begin{equation}
\begin{array}{rcl}
W_{1,2} &=& \begin{pmatrix}
\mathbf{N}_{11} & \mathbf{N}_{12} & \mathbf{N}_{13} & \cdots & \mathbf{N}_{1n} \\
\mathbf{N}_{21} & \mathbf{N}_{22} & \mathbf{N}_{23} & \cdots & \mathbf{N}_{2n} \\
\mathbf{N}_{31} & \mathbf{N}_{32} & \mathbf{N}_{33}& \cdots & \mathbf{N}_{3n} \\
\vdots & \vdots & \vdots & \ddots & \vdots \\
\mathbf{N}_{n1} & \mathbf{N}_{n2} & \mathbf{N}_{n3} & \cdots & \mathbf{N}_{nn} \\
\end{pmatrix} =e^{i(\theta_2-\theta_1)}(U\otimes U^\dagger)(P_{1,2}\otimes P_{1,2})W_{1,2}(P_{1,2}\otimes P_{1,2})(U\otimes U^\dagger)^{\dagger}\\
&{=}& e^{i(\theta_2-\theta_1)} \begin{pmatrix}
U^\dagger P_{1,2}\mathbf{N}_{11}P_{1,2}U & e^{i(\theta_2-\theta_1)}U^\dagger P_{1,2}\mathbf{N}_{12}P_{1,2}U & e^{i\theta_2}U^\dagger P_{1,2}\mathbf{N}_{13}P_{1,2}U & \cdots & e^{i\theta_2}U^\dagger P_{1,2}\mathbf{N}_{1n}P_{1,2}U \\
e^{i(\theta_1-\theta_2)}U^\dagger P_{1,2}\mathbf{N}_{21}P_{1,2}U & U^\dagger P_{1,2}\mathbf{N}_{22}P_{1,2}U & e^{i\theta_1}U^\dagger P_{1,2}\mathbf{N}_{23}P_{1,2}U & \cdots & e^{i\theta_1}U^\dagger P_{1,2}\mathbf{N}_{2n}P_{1,2}U \\ 
e^{-i\theta_2}U^\dagger P_{1,2}\mathbf{N}_{31}P_{1,2}U & e^{-i\theta_1}U^\dagger P_{1,2}\mathbf{N}_{32}P_{1,2}U & U^\dagger P_{1,2}\mathbf{N}_{33}P_{1,2}U & \cdots & U^\dagger P_{1,2}\mathbf{N}_{3n}P_{1,2}U \\
\vdots & \vdots & \vdots & \ddots & \vdots \\
e^{-i\theta_2}U^\dagger P_{1,2}\mathbf{N}_{n1}P_{1,2}U & e^{-i\theta_1}U^\dagger P_{1,2}\mathbf{N}_{n2}P_{1,2}U & U^\dagger P_{1,2}\mathbf{N}_{n3}P_{1,2}U & \cdots & U^\dagger P_{1,2}\mathbf{N}_{nn}P_{1,2}U \\
\end{pmatrix}
\end{array}.
\label{e12-3d}
\end{equation}
Each \( \mathbf{N}_{ij} \) is the solution to the following matrix equation for \( A=(a_{ij})_{i,j=1}^n \):
\begin{equation}
    A=kU^\dagger P_{12}AP_{12}U=k\begin{pmatrix}
a_{11} & e^{i(\theta_2-\theta_1)}a_{12} & e^{-i\theta_1}a_{13} & \cdots & e^{-i\theta_1}a_{1n} \\
e^{i(\theta_1-\theta_2)}a_{21} & a_{22} & e^{-i\theta_2}a_{23} & \cdots & e^{-i\theta_2}a_{2n} \\
e^{i\theta_1}a_{31} & e^{i\theta_2}a_{32} & a_{33} & \cdots & a_{3n} \\
\vdots & \vdots & \vdots & \ddots & \vdots \\
e^{i\theta_1}a_{n1} & e^{i\theta_2}a_{n2} & a_{n3} & \cdots & a_{nn} \\
\end{pmatrix},k\in \mathbf{C}.
\label{A10}
\end{equation}
Note that in $W_{1,2}$, the number $k = \begin{cases} 
e^{i(\theta_2-\theta_1)} &(i,j)=(1,1)\text{ and }(2,2) \text{ and } i\geq 3,j\geq 3 \\
e^{2i(\theta_2-\theta_1)}& i=1,j=2,\\
1& i=2,j=1,\\
e^{i(2\theta_2-\theta_1)}& i=1,j\geq 3,\\
e^{i\theta_2}& i=2,j\geq 3,\\
e^{-i\theta_1}& i\geq3,j=1,\\
e^{i(\theta_2-2\theta_1)}& i\geq3,j=2,\\
\end{cases}$ we have 
$\mathbf{N}_{i,j}= \begin{cases} 
k_{ij}E_{21} &(i,j)=(1,1)\text{ and }(2,2) \text{ and } i\geq 3,j\geq 3 \\
\mathbf{O}& i=1,j\geq2 \text{ and } i\geq3,j=2\\
X_{21}& i=2,j=1,\\
Y_{ij}& i=2,j\geq 3,\\
Z_{ij}& i\geq3,j=1,\\
\end{cases}$
according to equation \eqref{A10} and the arbitrariness of $\theta_1,\theta_2$. The elements from the third to the \( n \)-th in the second row of \( Y_{ij} \) are nonzero, and the elements from the third to the \( n \)-th in the first column of \( Z_{ij} \) are nonzero.Meanwhile, \( X_{2,1}  = diag\{m_{11},m_{22}\begin{pmatrix}
        m_{33}& \cdots& m_{3n}  \\
        \vdots& \vdots & \vdots\\
        m_{n3}&  \cdots& m_{nn}\\
    \end{pmatrix}\} \). 
Now let $U=\begin{pmatrix}
           &  &  e^{i\theta_2} &   & \\
        &  1  &   & & \\
         e^{i\theta_1} &   & &  & \\
          &   &   &  \ddots&\\
          &   &   &  &1
    \end{pmatrix}$, similarly we have
\begin{equation}
\begin{array}{rcl}
W_{1,2} &=&  e^{i\theta_1}(P_{13}\otimes P_{13})(U\otimes U^\dagger)W_{1,2}(U\otimes U^\dagger)^{\dagger}(P_{13}\otimes P_{13})\\
&{=}& e^{i\theta_1} \begin{pmatrix}
P_{1,3}U^\dagger \mathbf{N}_{11}UP_{1,3} & e^{i\theta_1}P_{1,3}U^\dagger \mathbf{N}_{12}UP_{1,3} & e^{i(\theta_1-\theta_2)}P_{1,3}U^\dagger \mathbf{N}_{13}UP_{1,3} & \cdots & e^{i\theta_1}P_{1,3}U^\dagger \mathbf{N}_{1n}UP_{1,3} \\
e^{-i\theta_1}P_{1,3}U^\dagger \mathbf{N}_{21}UP_{1,3} & P_{1,3}U^\dagger \mathbf{N}_{22}UP_{1,3} & e^{-i\theta_2}P_{1,3}U^\dagger \mathbf{N}_{23}UP_{1,3} & \cdots & P_{1,3}U^\dagger \mathbf{N}_{2n}UP_{1,3} \\ 
e^{i(\theta_2-\theta_1)}P_{1,3}U^\dagger \mathbf{N}_{31}UP_{1,3} & e^{i\theta_2}P_{1,3}U^\dagger \mathbf{N}_{32}UP_{1,3} & P_{1,3}U^\dagger \mathbf{N}_{33}UP_{1,3} & \cdots & e^{i\theta_2}P_{1,3}U^\dagger \mathbf{N}_{3n}UP_{1,3} \\
\vdots & \vdots & \vdots & \ddots & \vdots \\
e^{-i\theta_1}P_{1,3}U^\dagger \mathbf{N}_{n1}UP_{1,3} & P_{1,3}U^\dagger \mathbf{N}_{n2}UP_{1,3} & e^{-i\theta_2}P_{1,3}U^\dagger \mathbf{N}_{n3}UP_{1,3} & \cdots & P_{1,3}U^\dagger \mathbf{N}_{nn}UP_{1,3}  \\
\end{pmatrix}
\end{array}.
\label{e12-3d}
\end{equation}
Each \( \mathbf{N}_{ij} \) is the solution to the following matrix equation for \( C=(c_{ij})_{i,j=1}^n \):
\begin{equation}
    C=kP_{1,3}U^\dagger CUP_{1,3}=\begin{pmatrix}
m_{11} & e^{-i\theta_2}m_{12} & e^{i(\theta_1-\theta_2)}m_{13} & \cdots & e^{-i\theta_2}m_{1n} \\
e^{i\theta_2}m_{21} & m_{22} & e^{i\theta_1}m_{23} & \cdots & m_{2n} \\
e^{i(\theta_2-\theta_1)}m_{31} & e^{-i\theta_1}m_{32} & m_{33} & \cdots & e^{-i\theta_1}m_{3n} \\
\vdots & \vdots & \vdots & \ddots & \vdots \\
e^{i\theta_2}m_{n1} &a_{n2} & e^{i\theta_1}m_{n3} & \cdots & m_{nn} \\
\end{pmatrix},k\in \mathbf{C}.
\label{A}
\end{equation}
We have $k = \begin{cases} 
e^{i\theta_1} & (i\neq 1,3\text{ and }\\&  j\neq 1,3)\text{ and }i= j=1,3 \\
e^{i(2\theta_1-\theta_2)}& i=1,j=3,\\
e^{i\theta_2}& i=3,j=1,\\
1& j=1,i\neq 1,3,\\
e^{i(\theta_1-\theta_2)}& j=3,i\neq 1,3,\\
e^{i2\theta_1}& i=1,j\neq 1,3,\\
e^{i(\theta_1+\theta_2)}& i=3,j\neq 1,3,\\
\end{cases}$ and 
$\mathbf{N}_{ij} = \begin{cases} 
S_{ij} & (i\neq 1,3\text{ and }\\&  j\neq 1,3)\text{ and }i= j=1,3 \\
\mathbf{O}& i=1,j=3,\\
T_{ij}& i=3,j=1,\\
G_{ij}& j=1,i\neq 1,3,\\
Q_{ij}& j=3,i\neq 1,3,\\
\mathbf{O}& i=1,j\neq 1,3,\\
\mathbf{O}& i=3,j\neq 1,3,\\
\end{cases}$
The elements of \( S_{ij} \) that may be nonzero are only those in the third row except for the first and third elements. The elements of \( T_{ij} \) that may be nonzero are only those in the first row except for the first and third elements. The elements of \( G_{ij} \) that may be nonzero include the (1,1) element, the (3,3) element, and any element not in the first or third row or column. The element of \( Q_{ij} \) that may be nonzero is only the (3,1) element. By comparing the two expressions for \( b_{ij} \), it is easy to see that only \( \mathbf{N}_{21} \) is a nonzero matrix block, and with \( \mathbf{N}_{2,1} =  \begin{pmatrix}
        m_{11}& &   \\
       & m_{22} & \\
        &  & m_{33}\\
    \end{pmatrix}\oplus\begin{pmatrix}
        m_{44}& \cdots& m_{4n}  \\
        \vdots& \vdots & \vdots\\
        m_{n4}&  \cdots& m_{nn}\\
    \end{pmatrix}\} \).
\end{proof}
\begin{lemma}
    The (1,1) $n^2\times n^2$ matrix block \( W_{1,1} \) can be reduced to $diag\{\mathbf{M}_{1,1},\mathbf{M}_{2,2},\cdots,\mathbf{M}_{2,2}\}$,where the $n\times n$ matrices $\mathbf{M}_{1,1}=diag\{a,t,\cdots,t\},\mathbf{M}_{2,2}=diag\{d,c,\cdots,c\}$ .
\label{3d-diag}
\end{lemma}
\begin{proof}
    Let \( U = \operatorname{diag}\{e^{i\theta_1}, \cdots, e^{i\theta_n}\} \). We have
    \( W_{r,r} = (U \otimes U^\dagger) W_{r,r} (U \otimes U^\dagger)^\dagger,r=1,\cdots,n\) We assume that \( W_{r,r} = \{ \mathbf{M}^{(r)}_{ij} \} \). By simplification, we have
    $$e^{(\theta_i - \theta_j)} U^\dagger \mathbf{M}^{(r)}_{ij} U = \mathbf{M}^{(r)}_{ij},i=1,\cdots,n,j=1,\cdots,n.$$ 
This means that when $i=j$, \( U^\dagger \mathbf{M}^{(r)}_{i,i} U = \mathbf{M}^{(r)}_{i,i} \), so \( \forall \mathbf{M}^{(r)}_{i,i} \) is a diagonal matrix. For \( \forall i \neq j \), let \( \mathbf{M}^{(r)}_{ij} = (m^{(r)}_{kl})_{k,l=1}^n \), then
    \[
    e^{(\theta_i - \theta_j)} e^{(\theta_l - \theta_k)} m^{(r)}_{kl} = m^{(r)}_{kl},k=1,\cdots,n,l=1,\cdots,n.
    \]
    so the nonzero element in \( \mathbf{M}^{(r)}_{ij} \) is the \( (i,j) \)-th element, $r=1,\cdots,n$ So we can assume that $\mathbf{M}^{(r)}_{ij}=k_{ij}^{(r)}E_{i,j},\forall i\neq j$. 
    
Now let $U=\begin{pmatrix}
           & e^{i\theta_2} &   &   & \\
        e^{i\theta_1}&    &   & & \\
          &   & 1&  & \\
          &   &   &  \ddots&\\
          &   &   &  &1
    \end{pmatrix}$, we have $(U\otimes U^\dagger)W_{22}(U\otimes U^\dagger)^\dagger=W_{11}$. Since \( W_{r,r} = \{ \mathbf{M}^{(r)}_{ij} \} \), we have 
$$\mathbf{M}_{12}^{(1)}=e^{i(\theta_2-\theta_1)}U^\dagger \mathbf{M}_{21}^{(2)}U,M_{21}^{(1)}=e^{i(\theta_2-\theta_1)}U^\dagger \mathbf{M}_{12}^{(2)}U.$$
After simplification,we have $e^{2i(\theta_2-\theta_1)}k_{21}^{(2)}=k_{12}^{(1)},e^{2i(\theta_1-\theta_2)}k_{12}^{(2)}=k_{21}^{(1)}$. Due to the arbitariness of $\theta_1,\theta_2$,$k_{21}^{(2)}=k_{12}^{(1)}=k_{12}^{(2)}=k_{21}^{(1)}=0,\mathbf{M}_{12}=\mathbf{M}_{21}=\mathbf{O}$. Similarily, we let $U=P_{i,j}(\theta_1,\theta_2)$, where $P_{i,j}(\theta_1,\theta_2)$ is obtained by multiplying the $i-th$ column of $P_{i,j}$ by $e^{i\theta_1}$
and multiplying the $j-th$ column by $e^{i\theta_2}(i\neq j)$. Using the same technique, we can get $\mathbf{M}_{ij}^{(1)}=\mathbf{M}_{ji}^{(1)}=\mathbf{O}$.By iterating over all \(i, j\), we can obtain \(\mathbf{M}_{k,l}=\mathbf{O}(k\neq l)\). On the other hand, we have $$P_{ij}\mathbf{M}_{1,1}P_{ij}=\mathbf{M}_{1,1}(i,j\geq 2).$$
So we can set $\mathbf{M}_{1,1}=diag\{a,t,\cdots,t\}$. Similarly we have $P_{ij}\mathbf{M}_{2,2}P_{ij}=\mathbf{M}_{2,2}(i,j\geq 3)$,So we can set $\mathbf{M}_{2,2}=diag\{d,e,c,\cdots,c\}$ Inductively, the (1,1)-th element of $\mathbf{M}_{i,i}$ is $d,$the (i,i)-th element is $e,$ and other diagonal elements are $c$. 

Now let $U=diag\{1,\begin{pmatrix}
        \frac{1}{\sqrt 2}& \frac{1}{\sqrt 2} \\
        \frac{1}{\sqrt 2}& -\frac{1}{\sqrt 2}\\
    \end{pmatrix},1,\cdots,1\}$,we have $(U\otimes U^\dagger)W_{11}(U\otimes U^\dagger)^\dagger=W_{11}$. Specifically, consider the (2,2) $n\times n$ matrix block, we have $\frac{1}{2}U^\dagger(\mathbf{M}_{22}+\mathbf{M}_{33})U=\mathbf{M}_{22}$, which indicates that $e=c$, $\mathbf{M}_{ii}=\mathbf{M}_{jj}, \forall i,j\geq 2$. Then $W_{1,1}$ can be simiplified as $diag\{\mathbf{M}_{1,1},\mathbf{M}_{2,2},\cdots,\mathbf{M}_{2,2}\}$.
\end{proof}
With the results of \( W_{1,1} \) and \( W_{1,2} \) obtained above, we further connect them through a special matrix, ultimately arriving at the following theorem.
\begin{theorem}
\label{1112=0}
    $W_{1,1}=\begin{pmatrix}
        x+y& &   \\
       & x & \\
        &  & \ddots\\
        &  & &x\\
    \end{pmatrix}\otimes I_n=\tilde{W}_{1,1},W_{1,2}=\begin{pmatrix}
        \mathbf{O}& \mathbf{O}&\cdots&\mathbf{O} \\
        yI_n& \mathbf{O}&\cdots&\mathbf{O}\\
        \vdots& \vdots & \ddots&\vdots\\
        \mathbf{O}& \mathbf{O}&\cdots&\mathbf{O}
    \end{pmatrix}=\tilde{W}_{1,2}.$
\end{theorem}
\begin{proof}
By setting $U=P_{12}$, we have
$$(P_{12}\otimes P_{12})W_{2,2}(P_{12}\otimes P_{12})=W_{1,1},(P_{12}\otimes P_{12})W_{1,2}(P_{12}\otimes P_{12})=W_{2,1}.$$
We consider the matrices $\mathbf{N}_{21},\mathbf{M}_{11},\mathbf{M}_{22}$ is identical to the form in Lemmas \ref{3d-nondiag} and \ref{3d-diag}.
We denote $\mathbf{\tilde{M}}_{11}=diag\{t,a,t,\cdots,t\}=P_{12}\mathbf{M}_{11}P_{12}, \mathbf{\tilde{M}}_{22}=diag\{c,d,c,\cdots,c\}=P_{12}\mathbf{M}_{22}P_{12},\tilde{N}_{21}=P_{12}\mathbf{N}_{21}P_{12}$. We have
$$W_{2,1}=\begin{pmatrix}
        \mathbf{O}& \mathbf{\tilde{N}}_{21}&\cdots&\mathbf{O} \\
        \mathbf{O}& \mathbf{O}&\cdots&\mathbf{O}\\
        \vdots& \vdots & \ddots&\vdots\\
        \mathbf{O}& \mathbf{O}&\cdots&\mathbf{O}
    \end{pmatrix},W_{2,2}=\begin{pmatrix}
        \mathbf{\tilde{M}}_{22}& && &\\
        & \mathbf{\tilde{M}}_{11}&&&\\
        &  & \mathbf{\tilde{M}}_{22}&&\\
        & &&\ddots&\\
        & &&&\mathbf{\tilde{M}}_{22}\\
    \end{pmatrix}$$
Now let $U=diag\{\begin{pmatrix}
        \frac{1}{\sqrt 2}& \frac{i}{\sqrt 2} \\
        -\frac{1}{\sqrt 2}& \frac{i}{\sqrt 2}\\
    \end{pmatrix},1,\cdots,1\}$,we have the following matrix equation systems after simplification.
\begin{equation}
    U^\dagger(-\mathbf{M}_{11} + \mathbf{M}_{22} + \tilde{\mathbf{N}}_{21} - \mathbf{N}_{21} - \tilde{\mathbf{M}}_{22} + \tilde{\mathbf{M}}_{11})U = 0, 
    \label{eq:1} 
\end{equation}
\begin{equation}
    U^\dagger(-\mathbf{M}_{11} + \mathbf{M}_{22} - \tilde{\mathbf{N}}_{21} + \mathbf{N}_{21} - \tilde{\mathbf{M}}_{22} + \tilde{\mathbf{M}}_{11})U = 0 ,
    \label{eq:2} 
\end{equation}
\begin{equation}
    U^\dagger(\mathbf{M}_{11} + \mathbf{M}_{22} + \tilde{\mathbf{N}}_{21} + \mathbf{N}_{21} + \tilde{\mathbf{M}}_{22} + \tilde{\mathbf{M}}_{11})U = 4\mathbf{M}_{11} ,\label{eq:3} 
\end{equation}
\begin{equation}
    U^\dagger(\mathbf{M}_{11} + \mathbf{M}_{22} - \tilde{\mathbf{N}}_{21} - \mathbf{N}_{21} + \tilde{\mathbf{M}}_{22} + \tilde{\mathbf{M}}_{11})U = 4\mathbf{M}_{22} .\label{eq:4}
\end{equation}
Summing up equations \eqref{eq:3} and \eqref{eq:4}, we have $$U^\dagger(\mathbf{M}_{11} + \mathbf{M}_{22} + \tilde{\mathbf{M}}_{22} + \tilde{\mathbf{M}}_{11})U=2(\mathbf{M}_{11} + \mathbf{M}_{22}).$$ Simplifying gives \( a + d = t + c \). Adding equations \eqref{eq:1} and \eqref{eq:2} , we have $$U^\dagger(\tilde{\mathbf{M}}_{11}-\mathbf{M}_{11} + \mathbf{M}_{22}-\tilde{\mathbf{M}}_{22}    )U=\mathbf{O}.$$ 
Simplification gives \( t + d = a + c \). Solving these equations jointly, we obtain \( t = a \), \( d = c \), so \( \mathbf{M}_{11} = a I_n \), and \( \mathbf{M}_{22} = c I_n \). We now discuss \( \mathbf{N}_{21} \) and \( \tilde{\mathbf{N}}_{21} \).  \eqref{eq:1} - \eqref{eq:2} yields $U^\dagger(\mathbf{N}_{21}-\tilde{\mathbf{N}}_{21})U=0$, and simplification gives \( m_{11} = m_{22} \),thus $\mathbf{N}_{21}-\tilde{\mathbf{N}}_{21}$. Meanwhile, since \( \mathbf{M}_{11} = a I_n \), and \( \mathbf{M}_{22} = c I_n \), \eqref{eq:3} - \eqref{eq:4} gives 
$$U^\dagger\mathbf{N}_{21}U=(a-c)I_n.$$
simplifying gives \( m_{ii} = m_{jj} = a - c \) for all \( i, j \); \( m_{ij} = 0 \) for all \( i \neq j \), that is, \( \mathbf{N}_{2,1} = (a - c) I_n \). Let $y=a-c,x=c$,we have $W_{1,1}=diag\{(x+y)I_n,xI_n,\cdots,xI_n\}=\begin{pmatrix}
        x+y& &   \\
       & x & \\
        &  & \ddots\\
        &  & &x\\
    \end{pmatrix}\otimes I_n=\tilde{W}_{1,1},W_{1,2}=\begin{pmatrix}
        \mathbf{O}& \mathbf{O}&\cdots&\mathbf{O} \\
        yI_n& \mathbf{O}&\cdots&\mathbf{O}\\
        \vdots& \vdots & \ddots&\vdots\\
        \mathbf{O}& \mathbf{O}&\cdots&\mathbf{O}
    \end{pmatrix}=\tilde{W}_{1,2}$ According to Lemma \ref{centrosymmetric}, we obtain $W_{1,1}=\tilde{W}_{1,1},W_{1,2}=\tilde{W}_{1,2}$
\end{proof}
Combining Theorem \ref{1112} and Theorem \ref{1112=0},we obtain the following main theorem.
\begin{theorem}
    Any matrix $W$that commutes with $U\otimes U\otimes U^\dagger$ can be expressed as the following form, where $U$ represents arbitary unitary matrices.
\begin{equation}
    W=xI_{n^3} + yF_n \otimes I_n, \quad F_n = \sum_{i, j=1}^{n} (|i\rangle \langle j|) \otimes (|j\rangle \langle i|).
\label{final}
\end{equation}
\end{theorem}
\section{Conclusion}
\label{sec:con}
In conclusion, this work establishes that for arbitrary unitary matrices $U$ and $V$, the only matrices commuting with $U \otimes U^{\dagger}$ and $U \otimes V$ are scalar multiples of the identity matrix. Furthermore, in the tripartite system case, any matrix commuting with $U \otimes U \otimes U^{\dagger}$ must necessarily adopt the form in equation \eqref{final}, where $F_n$ represents the canonical swap operator. This study provides a unified framework for analyzing high-dimensional matrix structures, which can facilitate advancements in quantum information theory and tensor decomposition algorithms. 

Future investigations may explore two promising directions: (1) the commutant structure of operators with hybrid tensor configurations like $U^{\otimes n-1} \otimes U^{\dagger}$, and (2) systematic classification of commuting operators under diverse combinations of $U$, $U^{\dagger}$, and their complex conjugations $U^*$in tripartite systems. Such extensions could deepen our understanding of operator algebras in multipartite quantum scenarios.

\section*{ACKNOWLEDGMENTS}
Authors were supported by the NNSF of China (Grant No. 12471427). 
\bibliography{ref}
\end{document}